\documentclass[a4paper,11pt]{article}
\pdfoutput = 1
\usepackage{jheppub}
\usepackage{amsmath,amssymb,amscd,braket,amsfonts}
\usepackage{color}
\usepackage[dvipsnames]{xcolor}
\usepackage{graphicx}
\usepackage{slashed}
\usepackage{amsthm}
\usepackage{mathtools}

\newtheorem{theorem}{Theorem}[]

\def\p{\partial}
\def\({\left(}
\def\){\right)}
\def\[{\left[}
\def\]{\right]}
\newcommand{\bea}{\begin{eqnarray}}
\newcommand{\eea}{\end{eqnarray}}
\newcommand{\be}{\begin{equation}}
\newcommand{\ee}{\end{equation}}
\newcommand{\ba}{\begin{align}}
\newcommand{\ea}{\end{align}}
\newcommand{\im}{\mathfrak{Im}}
\newcommand{\re}{\mathfrak{Re}}

\title{Non-modal effects in black hole perturbation theory: Transient Superradiance}
\date{2024}

\author[a]{Javier Carballo,}
\emailAdd{j.carballo@soton.ac.uk}

\author[b]{Christiana Pantelidou,}
\emailAdd{christiana.pantelidou@ucd.ie}

\author[a]{and Benjamin Withers}
\emailAdd{b.s.withers@soton.ac.uk}

\affiliation[a]{Mathematical Sciences and STAG Research Centre, University of Southampton, Highfield, Southampton SO17 1BJ, UK}
\affiliation[b]{School of Mathematics and Statistics, University College Dublin, Belfield, Dublin 4, Ireland}

\abstract{We study the non-modal stability of black hole spacetimes under linear perturbations. We show that large-amplitude growth can occur at finite time, despite asymptotic decay of linear perturbations. In the example presented, the physical mechanism is a transient form of superradiance, and is qualitatively similar to the transition to turbulence in Navier-Stokes shear flows. As part of the construction we provide a theorem for the positivity of QNM energies, and introduce a truncated-Hamiltonian approach to black hole pseudospectra which does not suffer from convergence issues.}

\begin{document}
\maketitle

\section{Introduction and main results} \label{intro}

The study of linear black hole perturbations is of interest from the point of view of gravitational wave observations \cite{LIGOScientific:2014pky, LIGOScientific:2016aoc} and strongly-coupled many body systems through AdS/CFT \cite{Maldacena:1997re, Witten:1998qj}. Black hole perturbations explore the dissipative nature of the black hole horizon, and are consequently governed by \emph{non-normal} operators. This technical feature brings certain technical challenges, such as a lack of orthogonality and completeness of eigenfunctions, but it also means that black holes should display a wealth of interesting physical phenomena that normal systems do not. This work explores these possibilities and presents one such new phenomenon: \emph{transient superradiance}.

In previous investigations of the non-normality of black hole linear operators in the literature, great emphasis has been placed on analysing the `stability' of the spectrum of eigenfunctions (quasinormal modes, QNMs), under the influence of small changes to the environment \cite{Nollert:1996rf, Nollert:1998ys} (see also \cite{Daghigh:2020jyk}), now undergoing a recent resurgence \cite{Jaramillo:2020tuu}. However, non-normal systems also exhibit important \emph{non-modal} dynamical phenomena, rooted in aspects of linear perturbations which are not spectral. One way to access such non-modal physics is through analysis of the \emph{pseudospectrum} \cite{TrefethenEmbree2005}, and indeed the idea of using the pseudospectrum as a window into non-modal dynamical physics for binary black hole mergers was discussed in \cite{Jaramillo:2022kuv} and for horizonless compact objects in \cite{Boyanov:2022ark}.

Instead, in this work, we turn to a direct time domain analysis of non-modal phenomena in black hole perturbation theory. This programme was initiated in \cite{Carballo:2024kbk}, where it was shown that linear perturbations could decay arbitrarily slowly, despite all QNMs exhibiting fast exponential decay.\footnote{Boundedness of black hole perturbations without mode decomposition was previously investigated in \cite{Dotti14, Dotti:2016cqy}.} Here we show that black holes with decaying QNMs can admit perturbations whose energies grow in time. Such linear perturbations will ultimately decay at asymptotic time, but there is a transient period of significant growth which may source nonlinear effects. Indeed, non-modal growth is an important ingredient in the study of the transition to turbulence in fluid dynamics, and black holes display a strikingly similar phenomenology.
To illustrate our results in the simplest possible context, we focus on  Reissner-Nordstr\"om (RN) - AdS$_4$ spacetime, linearly perturbed by a charged complex scalar field. This allows us to introduce and exploit a non-modal analogue of superradiant scattering, as we shall explain.

We take our complex scalar field $\psi$ to have mass $m^2L^2 = -2$. QNMs can be defined as plane-wave perturbations $\psi(\tau,z,\vec{x}) = z^2\chi(z) e^{-i \omega \tau + i \vec{k}\cdot\vec{x}}$ which are ingoing at the future event horizon and normalisable at the AdS boundary. Here $\tau$ labels hyperboloidal slices of the spacetime -- spacelike slices that pierce the future event horizon. The eigenvalue problem which determines the spectrum of modes $\omega_n$ is as follows,
\be
\mathcal{H}\,\begin{pmatrix} 
   \chi(z) \\
     -i\omega \chi(z) 
    \end{pmatrix} = \omega \begin{pmatrix} 
   \chi(z) \\
     -i\omega \chi(z) 
    \end{pmatrix}, \label{hamchi}
\ee
where $\mathcal{H}$ is a 2$\times$2 matrix and a second-order differential operator in $z$, given later in \eqref{eq:Hamiltonian}. Imposing regularity of $\chi$ at the AdS boundary, $z=0$, at the future horizon, $z=1$,  corresponds to QNM boundary conditions and quantises the spectrum -- this is shown in figure \ref{fig:spec_probe} for the probe limit (charge of the scalar $q\to\infty$), and later in figure \ref{fig:spec} for the finite $q$ case. As the black hole temperature is lowered, some QNMs move into the upper half-plane and the system becomes modally unstable, corresponding to a transition to the broken phase of the holographic superconductor \cite{Gubser:2008px, Hartnoll:2008vx, Hartnoll:2008kx}, otherwise interpreted as a superradiant instability for a single mode, whose energy grows in time.

\begin{figure}[h!]
\centering
\includegraphics[width=0.5\columnwidth]{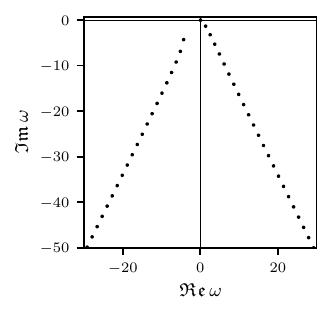}
\caption{The $\vec{k}=0$ QNM spectrum for charged scalar perturbations of the RN-AdS$_4$ black brane in the probe limit $q\to\infty$ at $\mu q=3.9$ where the system is modally stable. Finite $q$ results, beyond the probe limit, are given later in figure \ref{fig:spec}.}
\label{fig:spec_probe}
\end{figure}

We can assess the growth or decay of more general linear perturbations (beyond a single mode) by computing their energy, $E_\psi(\tau)$. The key result of this work is that $E_\psi(\tau)$ can grow, even when all QNMs are exponentially decaying. Mathematically, this occurs because $\mathcal{H}$ is non-normal with respect to the inner product associated to $E_\psi(\tau)$, and consequently the QNMs are not orthogonal to one another under this inner product.\footnote{There are several approaches to constructing orthogonality relations for QNMs in the literature \cite{Leaver, Green:2022htq, London:2023aeo}, however the relevant observable for us is the energy, and thus it is lack of orthogonality under this specific energy inner product which is of physical relevance.} Thus the energy of a sum of QNMs is not the sum of the energy of each QNM, allowing for a non-modal form of superradiance to occur even when each individual mode is superradiant stable. Earlier, it was proved that this allows for linear black hole perturbations which decay arbitrarily slowly \cite{Carballo:2024kbk}, despite each QNM decaying exponentially fast.

\begin{figure}[h!]
\centering
\includegraphics[width=0.5\columnwidth]{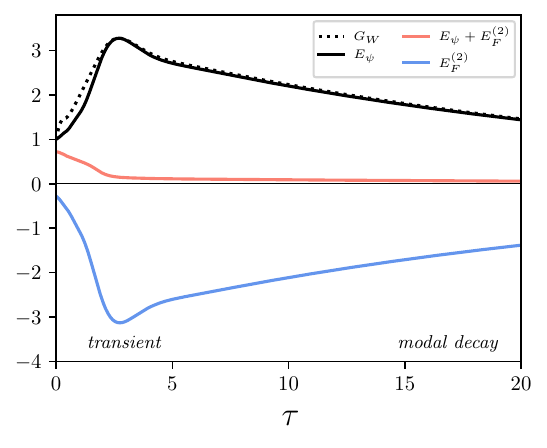}
\caption{Time evolution of the energy for linear perturbations to the holographic superconductor, $E_\psi$, demonstrating a period of transient growth despite modal stability. The mechanism is that of a transient form of superradiance, where energy is borrowed from the electric field, $E_{F}$. The example shown corresponds to the probe limit $q \to \infty$ with $\mu q = 3.9$, $\vec{k} = 0$, with initial data such that the system maximises $E_\psi$ at time $\tau_*=2.7$ within a subspace spanned by $M=10$ QNMs. The dotted line $G_W(\tau)$ gives a sharp upper bound on $E_\psi(\tau)$ for all possible initial data formed from 10 QNMs. Finite $q$ results, beyond the probe limit, are given later in figure \ref{fig:OptimalPert}.}
\label{fig:EEE}
\end{figure}

As an illustration of this phenomenon, the growth of $E_\psi(\tau)$ for a particular sum of QNMs can be seen in figure \ref{fig:EEE}. For simplicity of this introductory presentation, we have here removed backreaction by first taking the probe limit, $q\to\infty$ (finite $q$ is considered later). A microscopic interpretation of this growth is as follows. The RN-AdS$_4$ black hole has a radial electric field, and this encourages the classical wave analogue of pair production of $\pm$-charges outside the black hole. Like-charges are repelled from the black hole, forming a charged scalar cloud outside the horizon ($E_\psi$ increases), while opposite-charges are attracted to the black hole, depleting the strength of the bulk electric field (so that the energy associated to the electric field, $E_{F}$, decreases). The total energy $E = E_\psi + E_{F}$ can only decrease over time due to losses through the horizon. When the temperature is lowered beyond a critical value this is a runaway process leading to the formation of a hairy black hole. There is then a QNM which grows exponentially. At higher temperatures this process still occurs, but there are no growing QNMs and it is a transient phenomenon arising due to non-modal effects.

Full details of the derivation of these results are presented in section \ref{sec:details}. However, let us first review a similar situation in fluid dynamics and make a side-by-side comparison.

\section{The analogy with transients in plane-Poiseuille flow}\label{sec:oscompare}

As a point of comparison with the results outlined in section \ref{intro}, we consider a paradigmatic example of non-modal transients in fluid dynamics -- incompressible perturbations of the plane-Poiseuille flow. This shares much of the phenomenology of our black hole example, and we have summarised the analogous features where possible in table \ref{tab:analogy}. Full computations are given in appendix \ref{app:PPF}.

\begin{figure}[h!]
\centering
\includegraphics[width=0.5\columnwidth]{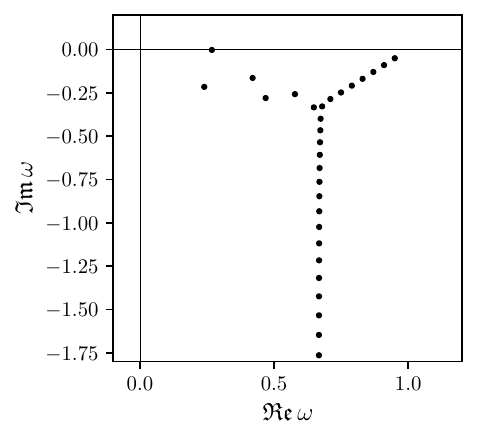}
\caption{The QNM spectrum for even perturbations to the plane-Poiseuille flow governed by the non-normal Orr-Sommerfeld operator. The choice of parameters is $\alpha = 1$, $Re=5000$ where the system is modally stable.}
\label{fig:OSspec}
\end{figure}

Plane-Poiseuille flow corresponds to a stationary, laminar solution to the Navier-Stokes equations with a flow in the $x$ direction with velocity profile $\vec{u} = (u^x, u^y, u^z) = (1-y^2,0,0)$ between two parallel $x-z$ plates at $y = \pm 1$, driven by a pressure gradient $\vec{\nabla} P = (-2\nu \rho, 0, 0)$ where $\nu$ is the viscosity and $\rho$ the density. Perturbing this flow by a stream function $\Phi(t,x,y)$ as follows, $\vec{u} = \left(1-y^2 + \partial_y\Phi, - \partial_x \Phi, 0\right)$ and decomposing into plane waves $\Phi = \phi(y)e^{-i\omega t + i \alpha x}$ gives the Orr-Sommerfeld equation,
\be
\alpha\,O_\text{OS}\,\phi(y) = \omega \phi(y). \label{OSeq}
\ee
Here $O_\text{OS}$ is the Orr-Sommerfeld differential operator \eqref{OSoperator}, characterised by the Reynolds number $Re = \nu^{-1}$. Boundary conditions correspond to $\phi(\pm 1) = \phi'(\pm 1) = 0$, which quantise the spectrum. The spectrum for plane-Poiseuille flow is shown in figure \ref{fig:OSspec}. As $Re$ is increased further, the system becomes modally unstable \cite{Orszag1971AccurateSO}.

The operator $O_\text{OS}$ is a non-normal operator and, as figure \ref{fig:OSspec} demonstrates, $\omega_n$ are complex. Here though, the origin of non-normality is bulk dissipation rather than loss through a boundary, due to the bulk viscous term, $\nu \nabla^2\vec{u}$. Transient growth in the energy of perturbations $E_\phi(t)$ is well-established \cite{Reddy93}, and we reproduce the phenomena here in figure \ref{fig:OSEEE}.

\begin{figure}[h!]
\centering
\includegraphics[width=0.5\columnwidth]{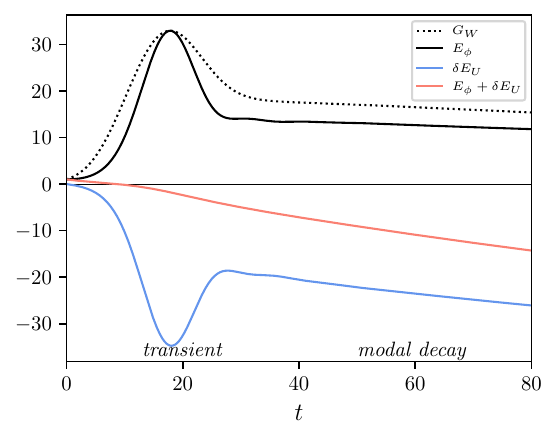}
\caption{$E_\phi(t)$, $\delta E_{U}(t)$ for an optimal linear perturbation to the plane-Poiseuille flow at $\alpha = 1$ and $Re = 5000$ for a subspace spanned by $M=30$ modes. The dotted line $G_W(t)$ gives a sharp upper bound on $E_\phi(t)$ for all possible initial data, and here initial data is chosen such that $E_\phi(t)$ reaches the maximum. The linear perturbation is characterised by an initial non-modal growth, even though the system is modally stable with all QNMs in the lower-half plane.}
\label{fig:OSEEE}
\end{figure}

The mechanism behind $E_\phi(t)$ growth is similar to black hole charge superradiance. Instead of pair creation of oppositely-charged particles in the presence of an electric field, in this case there is pair creation of oppositely-moving momentum modes in the presence of the background $x$-independent flow, $U(y) = 1-y^2$. $E_\phi(t)$ can then increase by non-modally borrowing energy from the mean flow energy, $E_{U}$. The net energy $E = E_\phi + E_{U}$ necessarily decreases due to viscous effects. This is shown in figure \ref{fig:OSEEE}. 

\begin{table}[h!]
\centering
\begin{tabular}{l|l}
\textbf{RN-AdS$_4$ black brane} & \textbf{Plane-Poiseulle flow}\\
\hline 
radial direction, $z$ & direction perpendicular to plates, $y$\\
$U(1)$ symmetry & translations in $x$\\
chemical potential deformation, $\mu$ & applied pressure gradient, $\partial_x P$\\
radial electric field, $E_z = 2 \mu$ & mean flow in $x$-direction, $U = 1-y^2$\\
non-normal Hamiltonian, $\mathcal{H}$ \eqref{hamchi} & non-normal Hamiltonian, $O_\text{OS}$ \eqref{OSeq}\\
charged scalar $\psi$ QNMs (fig. \ref{fig:spec}) & no-slip stream function perturbations $\Phi$ (fig. \ref{fig:OSspec})\\
superradiance (spontaneously broken $U(1)$) & turbulence (spontaneously broken $\partial_x$)\\
$\psi$ charge, $q$ & $\Phi$ wave number, $\alpha$\\
net energy loss through $\mathcal{H}^+$ & net energy loss via kinematic viscosity $\nu$ \eqref{dEviscous}
\end{tabular}
\caption{\label{tab:analogy} Summary of the coarse analogy between transient superradiance phenomena discussed in this work, and transient effects in perturbations to plane-Poiseulle flow.}
\end{table}

In this section we have drawn an analogy between plane-Poiseulle flow and the perturbations of charged black branes -- see table \ref{tab:analogy} for a summary. There are many differences between these systems, but the phenomenology is similar.  Part of this is because the mechanisms behind non-modal transients are similar, but also likely because of the connection between breaking $U(1)$ symmetry (in the case of charge superradiance) and breaking spacetime symmetries (in the case of turbulence and rotational superradiance) through dimensional reduction. It would be interesting to make this link precise within a concrete example.

\section{Method and further details} \label{sec:details}

\subsection{The bulk model}
We consider the four-dimensional Einstein-Maxwell action coupled to a complex scalar field, with charge $q$, given by

\be
S = \int d^4x \sqrt{-g}\left( R + 6 - \frac{1}{4}F^2 - |D\psi|^2 + 2|\psi|^2\right), \label{HHHaction}
\ee
corresponding to the holographic superconductor~\cite{Hartnoll:2008kx}. In \eqref{HHHaction}, $F=dA$, $D=\nabla - i q A$ and we have set $16 \pi G = 1$ and fixed the cosmological constant to be $\Lambda =-3$ (AdS radius $L=1$). The equations of motion associated with the above action are given by
\bea
\nabla_\mu F^{\mu\nu} &=& J^\nu\,,\nonumber\\
G_{\mu\nu} -3 g_{\mu\nu} &=& \frac{1}{2}T_{\mu\nu}\,,\nonumber\\
\left( D_\mu D^\mu + 2 \right ) \psi &=& 0\,,
\eea
where the conserved $U(1)$ current, $J^\mu$, and stress tensor, $T_{\mu\nu}$, are given by
\be
J^\mu = i q \left(\overline{\psi}D^\mu\psi - \psi \overline{D^\mu\psi} \right),
\ee
where $\overline\psi$ denotes the complex conjugate of $\psi$, and
\bea
\label{eq:StressTensor}
T_{\mu\nu} &=& T^\psi_{\mu\nu} + T^F_{\mu\nu}\,,\nonumber\\
T^\psi_{\mu\nu} &=& \overline{D_\mu\psi} D_\nu \psi +  D_\mu \psi\overline{D_\nu\psi} - g_{\mu\nu} |D\psi|^2 +g_{\mu\nu}(2|\psi|^2) \,,\nonumber\\
T^F_{\mu\nu} &=& F_{\mu\rho}F_{\nu}^{\phantom{\nu}\rho} - \frac{1}{4}g_{\mu\nu}F^2\,,
\eea
respectively, and they satisfy the local conservation equations
\bea\label{eq:conservation}
\nabla_\mu J^\mu &=& 0\,,\nonumber\\
\nabla_\mu T^{\mu\nu} &=& 0\,,\nonumber\\
\nabla_\mu (T_\psi)^{\mu}_{\phantom{\mu}\nu} &=& F_{\rho\nu}J^\rho\,,\nonumber\\
\nabla_\mu (T_F)^{\mu}_{\phantom{\mu}\nu} &=& -F_{\rho\nu}J^\rho\,.
\eea

The equations of motion then admit a unit-radius AdS$_4$ vacuum solution with $A =\psi = 0$, which is dual to a $d = 3$ CFT with an abelian global symmetry. In what follows, we are interested in placing the CFT at finite temperature $T$, with constant chemical potential $\mu$. The high temperature, spatially homogeneous and isotropic phase is described by the planar, electrically charged RN-AdS$_4$ black brane solution, which in Poincar\'e coordinates, $(t,r)$, takes the form
\bea
\label{eq:RN}
ds^2 &=& -f(r)dt^2 + \frac{dr^2}{f(r)} + r^2(dx_1^2 + dx_2^2),\nonumber\\
f(r) &=& r^2-\left(1 + \frac{\mu^2}{4}\right)\frac{1}{r} + \frac{\mu^2}{4r^2},\nonumber\\
A &=& \mu \left(1-\frac{1}{r}\right) dt\,, \quad \psi=0.
\eea
The black hole horizon is located at $r = 1$ in these coordinates\footnote{Throughout the paper, times and frequencies are given in units of $r_h=1$ which can be related to e.g. thermodynamic units by reinserting explicit factors of $r_h$ into \eqref{thermodynamics}.} and the associated thermodynamic quantities, namely the temperature $T$, entropy density $s$, charge density $\rho$, energy density $\epsilon$ and pressure $P$, are given by
\bea
T = \frac{3-\frac{\mu^2}{4}}{4\pi}, \quad s = 4\pi, \quad \rho = \mu, \quad \varepsilon = 2 + \frac{\mu^2}{2}, \quad P = 1 + \frac{\mu^2}{4}. \label{thermodynamics}
\eea

In this model the black brane \eqref{eq:RN} is unstable below some critical temperature~\cite{Denef:2009tp}, at any value of $q$. In the canonical ensemble the thermodynamically preferred black hole at low temperatures has non-vanishing charged scalar hair, describing a superfluid phase in the dual CFT. To obtain the critical temperature for this transition (or equivalently, critical chemical potential $\mu_c$), one looks for zero modes around \eqref{eq:RN}, that is, QNMs with $\omega = 0$.\footnote{Note here this requires a numerical solution, but for RN-AdS$_5$ the critical temperature can be determined analytically \cite{Arnaudo:2024sen}.} For later convenience, let us note that when $q=1$ we have $\mu_c\simeq 2.98$.

But of greater relevance for this work, another way to get the critical temperature is by considering (modal) energy growth. By studying linear perturbations of the system formed from a single QNM, it can be shown that energy can be removed from the system through superradiance if the frequency, $\omega$, of the corresponding QNM satisfies the condition
\be
\omega \in \left\{\tilde{\omega}\in\mathbb{C} \;\Big|\; \left(\re{\,\tilde{\omega}} + \frac{\mu q}{2}\right)^2 + (\im{\,\tilde{\omega}})^2 < \left(\frac{\mu q}{2}\right)^2\right\}. \label{superradcircle}
\ee
The critical temperature is then the highest one where a QNM satisfying \eqref{superradcircle} first exists. However, as we will see in section \ref{sec:positivity_thrm}, perturbations of the system composed of more than one QNMs allow for energy growth even above this critical temperature.

\subsection{Hyperboloidal coordinates}
The holographic superconductor is most commonly analysed using Poincar\'e coordinates $(t,r)$. In what follows we will instead use a hyperboloidal slicing: spacelike slices that pierce the future event horizon, rather than degenerate at the bifurcation point. The advantage of using this coordinate system is that it avoids the past horizon where QNMs are singular, and provides a way to track the amount of energy falling into the black hole over time.

Hyperboloidal coordinates \cite{Schmidt:1993rcx, Zenginoglu:2007jw, Dyatlov:2010hq, Bizon:2010mp,  Warnick:2013hba, PanossoMacedo:2018hab, Gajic:2019qdd, Bizon:2020qnd, PanossoMacedo:2023qzp} are obtained starting from Poincar\'e coordinates via the following coordinate transformation
\be
t=\tau-h(z)\,,\quad r=R(z)\,,
\ee
where the height function $h(z)$ bends the original Cauchy slice so that $\tau= \text{const.}$ corresponds to a hypersurface $\Sigma_\tau$ which penetrates the black hole horizon, and for convenience $R(z)$ is used to perform a spatial compactification. Here the new radial coordinate $z\in[0,1]$ is chosen such that $z = 1$ corresponds to the future horizon, while $z = 0$ corresponds to the conformal boundary of AdS. In this parametrisation, the metric and gauge field read
\bea
\label{eq:RNAdS4_hyperb}
ds^2&=&-\tilde{f}(z)\,d\tau^2+2\tilde{f}(z)h'(z)\,d\tau dz + \left (\frac{R'(z)^2}{\tilde{f}(z)} - \tilde{f}(z)h'(z)^2 \right )dz^2 + R(z)^2 d\Vec{x}^{\,2},\nonumber \\
A &=& \mu\left(1-{R(z)}^{-1}\right)\,d\tau - \mu\left(1-{R(z)}^{-1}\right) h'(z)\,dz\,, \quad  \quad\tilde{f}(z)=f[R(z)]. 
\eea
Our choice of coordinates is given by
\bea 
h(z)& =& \frac{1}{2\kappa}\log (1-z) - \frac{z_c^2}{2\kappa_c}\left(\log(z_c-z)-\log(z_c)\right)\,,\label{hchoice}\\
R(z)&=&1/z\,,
\eea
chosen so that the mode is ingoing at the Cauchy horizon, while ensuring that $h(0) = 0$ so that the asymptotic time is not adjusted. Here $z_c$ is the location of the Cauchy horizon, and $\kappa$, $\kappa_c$ are the surface gravities at the horizon and Cauchy horizon respectively
\bea
\kappa &=& \frac{12-\mu^2}{8},\nonumber\\
\kappa_c &=& \frac{2 z_c^4 \mu^2 -z_c^3 \mu^2 - 4z_c^3-8}{8 z_c},\\
z_c &=& { \frac{2 \sqrt[3]{54 \mu ^4+6 \left(\sqrt{3} \sqrt{27 \mu ^4+56 \mu ^2+48}+12\right) \mu ^2+64}+\frac{4\ 2^{2/3} \left(3 \mu ^2+4\right)}{\sqrt[3]{27 \mu ^4+3 \left(\sqrt{3} \sqrt{27 \mu ^4+56 \mu ^2+48}+12\right) \mu ^2+32}}+8}{6 \mu ^2}.}\nonumber
\eea
Demanding that $\Sigma_\tau$ intersects the future event horizon only requires the first term in \eqref{hchoice}. However, when approaching low temperatures, the influence of the Cauchy horizon increases and we require subsequent terms $\sim \log(z_c-z)$ to maintain a good spacelike slice. 
These are similar to the slices used in~\cite{Destounis:2021lum} which also contain such terms, and reduce to those used in \cite{Carballo:2024kbk} for Schwarzschild-AdS when $\mu \to 0$.

If we linearise the scalar field $\psi$ around the RN-AdS$_4$ background,
\be
\psi = \psi^{(1)}\epsilon + O(\epsilon)^2 \label{scalarlinear}
\ee
and decompose into plane waves as follows,
\be
\psi^{(1)}(\tau,z,\vec{x})=\int\frac{d^2\vec{k}}{(2\pi)^2}z^2\chi_{\vec{k}}(\tau,z)e^{i\vec{k}\cdot\vec{x}}\,, \label{scalarplanewave}
\ee
then the equation of motion for $\chi_{\vec{k}}$ is given by
\be
\label{eq:Hamiltonian}
i\partial_\tau \begin{pmatrix} 
   \chi_{\vec{k}} \\
     \partial_\tau \chi_{\vec{k}} 
    \end{pmatrix}=\mathcal{H}\,\begin{pmatrix} 
   \chi_{\vec{k}} \\
     \partial_\tau \chi_{\vec{k}} 
    \end{pmatrix}\,,\quad \mathcal{H}=\begin{pmatrix}
0 & i\\
\mathcal{L}_1& \mathcal{L}_2
\end{pmatrix}
\ee
where 
\bea
\mathcal{L}_1&=&\frac{-i}{-1+z^4\tilde{f}^2h'^2}\left[q^2 (-1 + z)^2 \mu^2 + 2 z^2 \tilde{f}^2 +\tilde{f}(2-\vec{k}^2z^2+2z^3\tilde{f}')+z^3 \tilde{f}\(z\tilde{f}'\p_z+\tilde{f}(4\p_z+z\p_{zz})\)\right]\,,\nonumber\\
\mathcal{L}_2&=&\frac{-i}{-1+z^4\tilde{f}^2h'^2}\left[-2iq(-1+z)\mu +z^4\tilde{f}\tilde{f}'h'+z^3\tilde{f}^2(4h'+zh'')+2z^4\tilde{f}^2 h'\partial_z\right]\,.
\eea
Note that, for convenience, in the decomposition above we have also removed a factor of $z^{\Delta}$, with $\Delta=2$ corresponding to the scaling dimension of the operator holographically dual to $\psi$ -- this ensures that regularity of $\chi_{\vec{k}}$ enforces the required normalisable behaviour at $z = 0$.

\subsection{Energy and charge}
The main observable of interest is the energy of the scalar field on a spacelike slice labelled by $\tau$, $\Sigma_\tau$, given by
\be
E_\psi \equiv \int_{\Sigma_\tau}(T_\psi)^{\mu}_{\phantom{z}\tau}n_{\mu}d\Sigma_\tau, \label{Epsidef}
\ee
where $n=\frac{-1}{\sqrt{-g^{\tau\tau}}}d\tau$ is the unit, future-directed normal to $\Sigma_\tau$.
This is conserved up to flux through the future horizon and through exchange with the gauge field, since,
\be
\partial_\tau E_\psi = \int d^2\vec{x} (T_\psi)^{z}_{\phantom{\tau}\tau}\big|_{z=1} - \int \sqrt{-g}F_{z \tau}J^z \, dz d^2\vec{x}.\label{Epsicons}\
\ee
On the other hand, the total energy $E$, as well as the charge $Q$, 
\bea
E &\equiv & \int_{\Sigma_\tau}T^{\mu}_{\phantom{z}\tau}n_{\mu}\,d\Sigma_\tau,\label{Edef}\\
Q &\equiv & \int_{\Sigma_\tau} J^{\mu} n_\mu\, d\Sigma_\tau,\label{Qdef}
\eea
are both conserved up to only boundary terms. In the special case of the probe limit -- where $q\to\infty$ so that backreaction is parameterically suppressed -- the additional energy comes only from the gauge field, and we denote this contribution by $E_F \equiv E - E_\psi$.

To proceed we linearise the scalar as in \eqref{scalarlinear} and evaluate $E_\psi$, $Q$ to $O(\epsilon)^2$. Since $T^\psi_{\mu\nu}$ and $J^\mu$ are quadratic in $\psi$, scalar perturbations to $O(\epsilon)$ are sufficient. With a plane wave decomposition of $\psi^{(1)}$ \eqref{scalarplanewave}, we have the following scalar contribution to energy,
\bea
\label{eq:Epsi}
E_{\psi}= \epsilon^2\int\frac{d^2\vec{k}}{(2\pi)^2}\int_0^1 dz \,&\big[&  w_E(z) \p_\tau \overline\chi_{\vec{k}} \p_\tau \chi_{\vec{k}} + p_E(z) \p_z \overline \chi_{\vec{k}} \p_z \chi_{\vec{k}} + q_E(z) \overline \chi_{\vec{k}} \chi_{\vec{k}}  \label{energyintegral}\\ 
&+& \left ( \alpha_1(z) \overline\chi_{\vec{k}} \p_\tau \chi_{\vec{k}} + \text{c.c}\right ) + \left ( \alpha_2(z) \overline\chi_{\vec{k}} \p_z \chi_{\vec{k}} + \text{c.c}\right )\big] + O(\epsilon)^4,\nonumber
\eea
and charge,
\be
Q =\epsilon^2\int \frac{d^2\Vec{k}}{(2\pi)^2}\int_{0}^{1} \left[ \left(w_Q(z) \overline\chi_{\vec{k}}\p_\tau\chi_{\vec{k}} + p_Q(z) \overline\chi_{\vec{k}}\p_z\chi_{\vec{k}} + q_Q(z) \overline\chi_{\vec{k}}\chi_{\vec{k}} \right)+ \text{c.c.} \right] dz + O(\epsilon)^4,\label{chargeintegral}
\ee
where the coefficient functions $w_E(z)$, $p_E(z)$, $q_E(z)$, $\alpha_1(z)$, $\alpha_2(z)$, $w_Q(z)$, $p_Q(z)$, $q_Q(z)$ are given in appendix \ref{app:wpq}. In what follows we will refer to $E_{\psi}$ and $Q$ as the above $O(\epsilon)^2$ pieces, with the formal parameter set to $\epsilon = 1$. 
On the other hand, $E$   begins at order $\epsilon^0$
\bea
E = \frac{\mu^2}{2}\text{vol}_2 + O(\epsilon)^2,
\eea
receiving corrections due to both the gauge field and metric at order $\epsilon^2$. In the probe limit in particular, one has $E_F \equiv E - E_\psi = \int_{\Sigma_\tau}(T_F)^{\mu}_{\phantom{z}\tau}n_{\mu}d\Sigma_\tau = \frac{\mu^2}{2}\text{vol}_2 + E_F^{(2)}\epsilon^2 + O(\epsilon)^4$.  

Finally, we consider the contribution to $E_\psi$ coming from a single $\vec{k}$-mode, $\chi_{\vec{k}}(\tau,z)$, and define the associated energy inner product,\footnote{From now on, we drop the $\vec{k}$ label: $\chi(\tau,z)$.}
\be
\label{eq:innerprod}
\langle \xi_1, \xi_2 \rangle \equiv \int_0^1 dz \,\, \xi_{1}^{*} \cdot \mathcal{G}\cdot \xi_2\,,
\ee
where we have introduced the notation $\xi(\tau,z) \equiv \begin{pmatrix} \chi(\tau,z)\\ \partial_\tau \chi(\tau,z) \end{pmatrix}$ -- here and throughout the text, $^*$ denotes the conjugate transpose -- and
\be
\mathcal{G}=\begin{pmatrix}
\overset{\leftarrow}{\p_z} \,p_E(z)\,\overset{\rightarrow}{\p_z}+q_E(z)+\alpha_2(z)\,\overset{\rightarrow}{\p_z}+\overset{\leftarrow}{\p_z}\,\overline \alpha_2(z) & \;\alpha_1(z)\\
\overline \alpha_1(z) &   \;w_E(z) 
\end{pmatrix},
\ee
such that $E_\psi[\xi] = \langle \xi,\xi\rangle$.

\begin{figure}[h!]
\centering
\includegraphics[width=0.5\columnwidth]{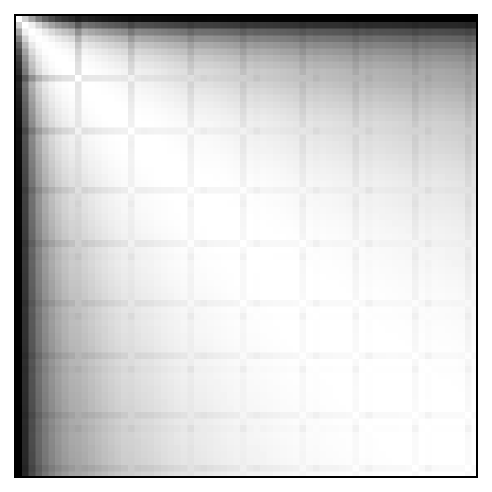}
\caption{Overlap $|\langle \tilde{\xi}_i, \tilde{\xi}_j \rangle|$  for QNMs $i$ and $j$, when ordered by their imaginary part, starting from the longest-lived mode at $i=j=1$ in the top left of the diagram. For diagonal elements, $i=j$, the overlap satisfies $\langle \tilde{\xi}_{i}, \tilde{\xi}_{j} \rangle=1$ (white). For off-diagonal elements, $i\ne j$,  $0<|\langle \tilde{\xi}_{i}, \tilde{\xi}_{j} \rangle|<1$ (grey, with lighter shades corresponding to larger values). The choice of parameters used is $q=1$, $\mu=2.9$, $\vec{k}=0$.}
\label{fig:array}
\end{figure}

Under the energy inner product \eqref{eq:innerprod}, QNMs are not orthogonal to one another. In particular, this means that the energy of a sum of QNMs is not the sum of their energies, and the additional terms allow for transient effects in the time evolution $E_\psi$. Figure \ref{fig:array} quantifies the overlap between pairs of QNMs for the spectrum shown in figure \ref{fig:spec}. Lighter colour corresponds to a stronger overlap. Specifically, for any two QNMs, labelled by $i$ and $j$, the figure shows the value of $|\langle \tilde{\xi}_i, \tilde{\xi}_j \rangle|$, where $\tilde{\xi}$ are the normalised eigenfunctions
\be
\tilde{\xi} = \frac{\xi(z)}{\sqrt{\left<\xi(z),\xi(z)\right>}}.
\ee
The visible grid-like structure corresponds to a lower overlap between different branches in the spectrum.

\subsection{Optimal perturbations and scalar energy growth curve} \label{sec:optimal}

In this subsection we briefly review the algorithm for calculating optimal perturbations -- those perturbations which maximise energy growth $E_\psi[\xi(\tau,z)]/E_\psi[\xi(0,z)]$ at a fiducial finite time $\tau$ (see~\cite{Carballo:2024kbk} for more details), and then apply it to our model. The algorithm is as follows:
\begin{enumerate}
\item Construct QNMs corresponding to perturbations of the charged scalar field. Specifically, we denote the eigenfunctions as $\xi_n(z)$ and eigenvalues $\omega_n$, with $n=1,2,\dots$ in order of decreasing imaginary part. Here, this step is performed numerically.
\item Select a finite set of QNMs consisting of the first $M$ modes, $\{\xi_n\}_{n=1}^M$. Let us denote the subspace of linear perturbations spanned by these QNMs as $W$, such that $\dim(W)=M$. 
\item Use the scalar field energy norm to  normalise the corresponding eigenfunctions 
\be
\tilde\xi_n=\frac{\xi_n(z)}{\sqrt{\langle \xi_n(z),\xi_n(z)\rangle}}\,.
\ee
\item Use the Gram-Schmidt method to construct an orthonormal set of functions $\{\zeta_n\}_{n=1}^M$ satisfying $\langle \zeta_i,\zeta_j\rangle=\delta_{ij}$. This process will also yield the change of basis matrix $U_W$ that corresponds to the projections between these two sets of functions, i.e
\be
U_W=
\begin{pmatrix}
\langle \zeta_1,\tilde\xi_1\rangle & \langle \zeta_1,\tilde\xi_2\rangle & \dots& \langle \zeta_1,\tilde\xi_M\rangle\\
0 & \langle \zeta_2,\tilde\xi_2\rangle & \dots& \langle \zeta_2,\tilde\xi_M\rangle\\
\vdots & \vdots & \ddots& \vdots\\
0 & 0 & \dots& \langle \zeta_M,\tilde\xi_M\rangle\\
\end{pmatrix}.
\ee
\item Obtain the matrix $H_W=U_W D_W U_W^{-1}$, where $D_W = \text{diag}(\omega_1, \omega_2, \dots, \omega_M )$.
\item Carry out a singular value decomposition for the $W$-projected time evolution operator, $e^{-i H_W\tau}$. Its maximum singular value squared computes the energy growth curve in $W$ \cite{Carballo:2024kbk}
\be
\label{eq:GW}
G_W(\tau)=\sup_{\xi(0,z)\in W}\frac{E_\psi[\xi(\tau,z)]}{E_\psi[\xi(0,z)]}=\|e^{-i H_W \tau}\|^2_2\,,
\ee
which captures the maximum possible $E_\psi$ at any given time $\tau\ge0$ within $W$. Here $\|\cdot\|_2$ denotes the usual $l^2$-norm induced from the Euclidean inner product $\langle \Vec{e}_1, \Vec{e}_2 \rangle_2 = \Vec{e}_{1}^{\,\,*}\Vec{e}_{2}$. On the other hand, its right principal singular vector gives a set of coefficients $\vec{d}$ such that
\be
\label{eq:expansion}
\xi(0,z)=\sum_{n=1}^M c_n \tilde\xi_n=\sum_{n=1}^M d_n \zeta_n
\ee
corresponds to optimal initial data maximising the energy, $E_\psi$ at time $\tau$. It should be clarified that the energy growth $G_W(\tau)$ is not the time evolution of the energy of a perturbation, $E_\psi[\xi(\tau,z)]$. Rather, there exists initial data $\xi(0,z) \in W$ for each $\tau\ge0$ such that $G_W(\tau)=E_\psi[\xi(\tau,z)]$. 
\end{enumerate}

The numerical implementation of the above algorithm makes use of Chebyshev spectral methods. The inner product is also discretised into $\langle \xi_1,\xi_2\rangle=\vec{\xi}_1^*\, G\,\vec{\xi}_2= (F \vec{\xi}_1)^*(F \vec{\xi}_2)$, where $G$ is a $2(N + 1) \times 2(N + 1)$ matrix and $F$ its Cholesky decomposition. Note that in the computation of $G$ one needs to include the quadrature weights coming from integrating over a Chebyshev grid. In addition, in order to minimise loss of accuracy resulting from approximating the integrand in \eqref{eq:innerprod} as a single Chebyshev expansion, we first compute $G$ on a grid of double resolution $2(N+1)$ and then interpolate down to $N+1$ grid points. Throughout the paper, unless otherwise stated, all numerical results are derived at $N=450$ and using $200$ digits of precision.

\begin{figure}[h!]
\centering
\includegraphics[width=0.5\columnwidth]{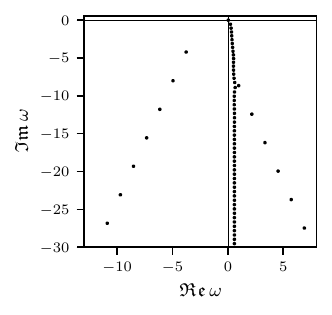}
\caption{The $\vec{k}=0$ QNM spectrum for $q=1$ charged scalar perturbations of the RN-AdS$_4$ black brane at $\mu=2.9$ where the system is modally stable ($\mu_c\simeq 2.98$ at $q=1$).}
\label{fig:spec}
\end{figure}

\begin{figure}[h!]
\centering
\includegraphics[width=1.0\columnwidth]{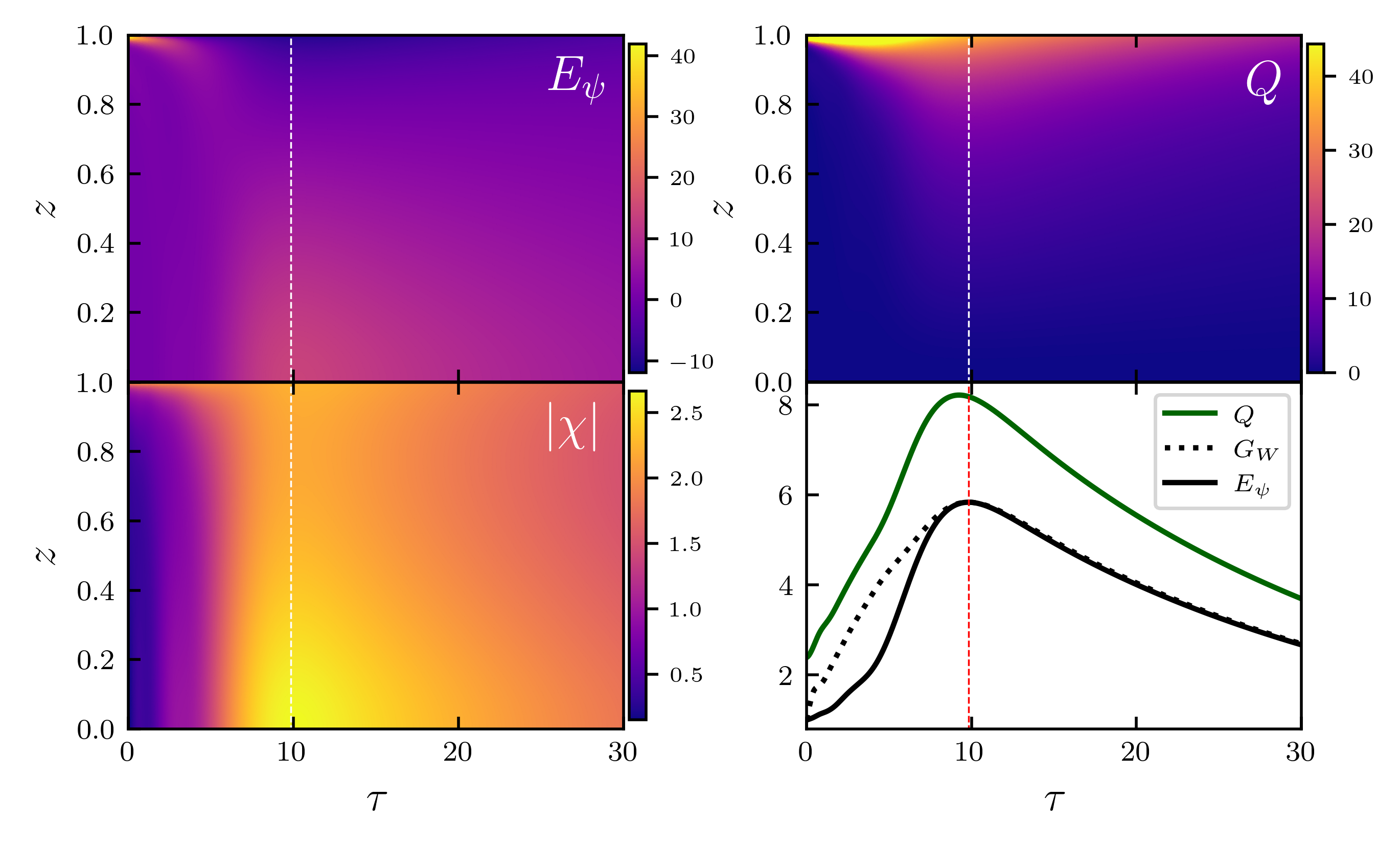}
\caption{Optimal perturbation in the subcritical regime, demonstrating a period of transient growth despite modal stability.  $E_\psi, |\chi|, Q$ densities as functions of $\tau$ and $z$, showing that the scalar field is localised close to the horizon. The maximum displayed value of the $Q$ density is $44.3$ for visualisation purposes. Integrating these on the hyperboloidal slice $\Sigma_\tau$, the last panel also shows the time evolution of 
$E_\psi,Q$ given optima, clearly demonstrating transient growth for a period of time, before it eventually decays modally. The dotted line $G_W(\tau)$ gives a sharp upper bound on $E_\psi$ for all possible initial data. The choice of parameters is $q=1$, $\mu=2.9$, $\vec{k}=0$ and $\tau_\ast=9.83$.}
\label{fig:OptimalPert}
\end{figure}

Let us now discuss our results. We focus on parameters $q=1, \mu=2.9$ and $\vec{k}=0$, corresponding to a subcritical region of the phase diagram ($\mu<\mu_c$) where all QNMs are decaying in time. The QNM spectrum for this choice of parameters is shown in figure \ref{fig:spec}. In this subcritical region, we can define the \emph{growth factor} in $W$ as the maximum possible $E_\psi[\xi(\tau,z)]$ over all $\tau$ for any choice of $\xi(0,z) \in W$. By definition of $G_W(\tau)$ \eqref{eq:GW}, the growth factor is simply $G_W(\tau_{\text{peak}})\equiv\sup_{\tau\geq0}G_W(\tau)$. Figure \ref{fig:OptimalPert} shows the optimal perturbation constructed out of $M=10$ modes for the subcritical holographic superconductor, chosen to maximise the energy growth at a time $\tau_\ast=\tau_{\text{peak}}=9.83$. With these parameters, there is a growth factor of $G_W(\tau_{\text{peak}})\simeq 5.83$, while soon after $\tau_\ast$ there is a transition to modal decay.

\begin{figure}[h!]
\centering
\includegraphics[width=0.6\columnwidth]{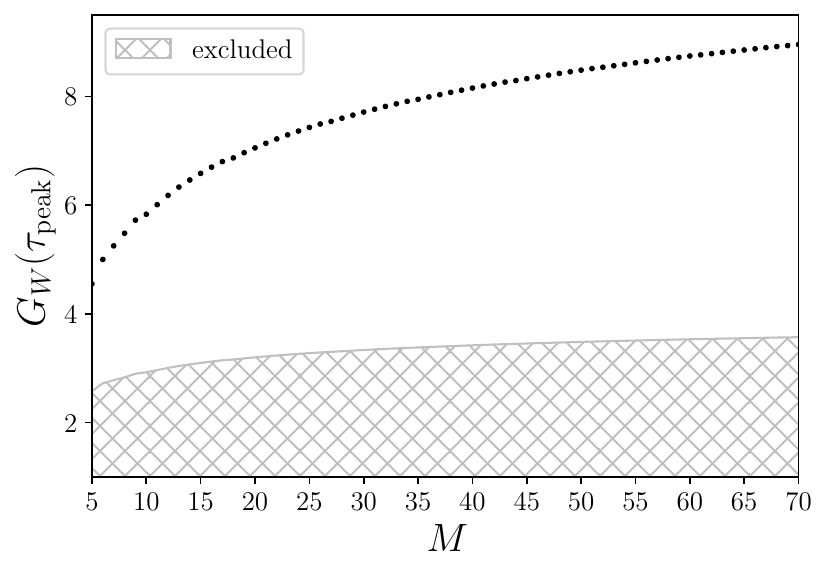}
\caption{Growth factor as a function of $M$. The choice of parameters is $q=1$, $\mu=2.9$, $\vec{k}=0$.}
\label{fig:GrowthFactor}
\end{figure}

The growth factor $G_W(\tau_{\text{peak}})$ does not however appear to be bounded with the number of modes. This is illustrated in figure \ref{fig:GrowthFactor} which shows the scalar energy growth factor, $G_W(\tau_{\text{peak}})$, as a function of $M=\dim(W)$. This may be an important consideration for the seeding of nonlinear effects in the continuum theory.\footnote{Note that in the Orr-Sommerfeld case the growth factor saturates at around $M\simeq 60$.}

\begin{figure}[h!]
\centering
\includegraphics[width=0.7\columnwidth]{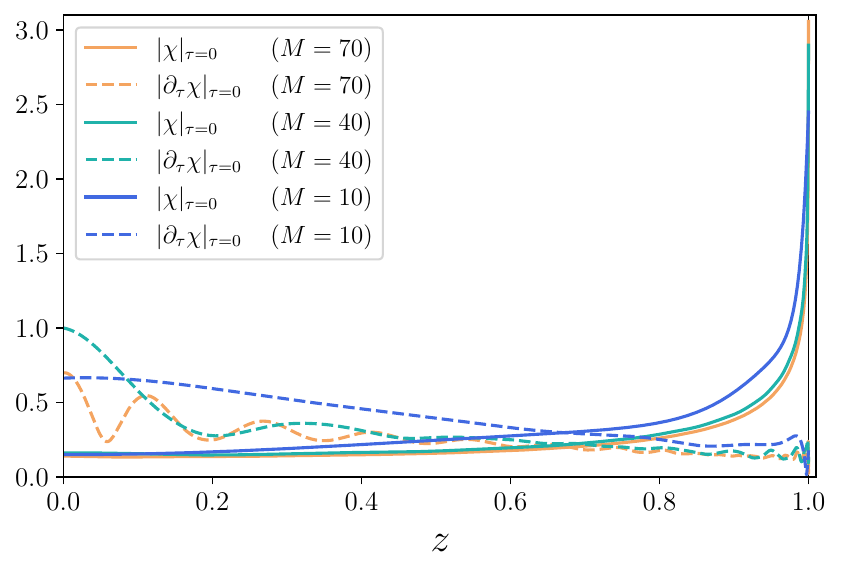}
\caption{The initial data corresponding to optimal perturbations of the RN-AdS$_4$ black brane, formed from a sum of the leading $M$ QNMs with $\vec{k}=0$. Each example of initial data shown has $E_\psi = 1$, which then increases transiently. Optimal perturbations are those which reach the maximum possible energy during this transient period. The choice of parameters is $q=1$, $\mu=2.9$.}
\label{fig:initdata}
\end{figure}

The corresponding optimal initial data is localised near the horizon, see figure \ref{fig:initdata}. This then expands to fill the spacetime before decaying via QNMs, as illustrated in the density plots of figure \ref{fig:OptimalPert}.\footnote{The energy and charge densities are computed by evaluating the integrands as shown in \eqref{Epsidef} and \eqref{Qdef}; note that the energy density differs from the integrand of \eqref{energyintegral} by total derivative terms.} There is an oscillatory behaviour in $|\p_\tau\chi|_{\tau=0}$ close to the horizon for larger $M$. This has an imprint on the densities $E_\psi, |\chi|, Q$: portions of these densities spread out of the horizon over slightly different time scales. On a figure similar to figure \ref{fig:OptimalPert}, this appears as a number of `beams' of density emanating from the horizon.

\begin{figure}[h!]
\centering
\includegraphics[width=0.7\columnwidth]{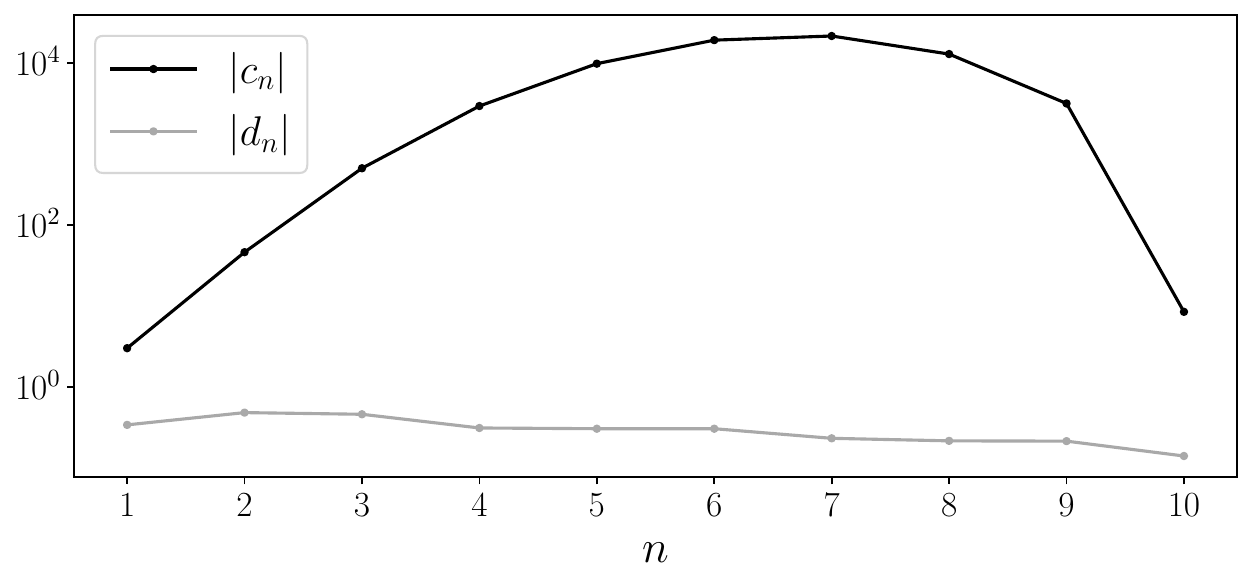}
\caption{The QNM coefficients, $c_n$, (black) in the optimal perturbation shown in figure \ref{fig:OptimalPert}, at the initial data surface $\tau = 0$ \eqref{eq:expansion}. Also shown are the coefficients in the orthonormal basis for $W$, $d_n$~(grey). Note $\sqrt{\vec{d}^{\,*}\vec{d}}=1$ and $\sqrt{\vec{c}^{\,*}\vec{c}}\simeq10^4$. The choice of parameters is $q=1$, $\mu=2.9$, $\vec{k}=0$ and $\tau_*=9.83$.}
\label{fig:OptimalPert_coeffs}
\end{figure}

Finally, the magnitude of the coefficients in the decomposition of the optimal perturbation is shown in figure \ref{fig:OptimalPert_coeffs}. While the coefficients in the orthonormal basis for $W$ are order one (the $d_n$), the coefficients in a sum of QNMs are large (the $c_n$).

\begin{figure}[h!]
\centering
\includegraphics[width=0.5\columnwidth]{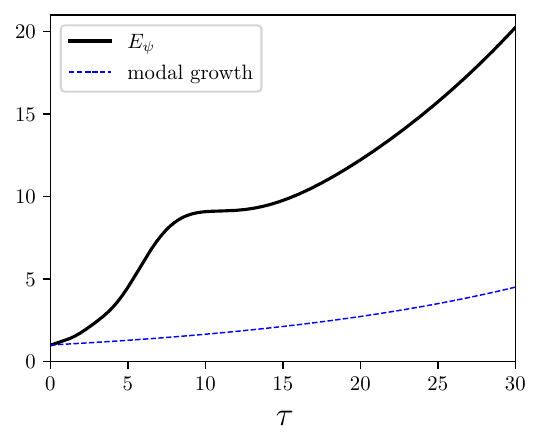}
\caption{Optimal perturbation, $E_\psi$, in the supercritical regime (black solid line). This demonstrates a sum of QNMs with a growth rate much faster than the single unstable QNM in the spectrum (blue dashed line). The choice of parameters is  $\mu= 3.1>\mu_c$, $q = 1$, $\vec{k}= 0$, $M=10$ and $\tau_\ast=7.5$.}
\label{fig:optimal_unstable}
\end{figure}

One can repeat the investigation in this section but at supercritical temperatures. For such values of $\mu,q$, the spectrum includes an exponentially growing mode. However, we find that non-modal optimal perturbations grow much faster than the QNM growth. Our results for the optimal perturbation in this case are shown in figure \ref{fig:optimal_unstable} for $\mu= 3.1>\mu_c$, $q = 1$, $\vec{k}= 0$, $M=10$ with target time $\tau_\ast = 7.5$.

\subsection{Truncated-Hamiltonian pseudospectrum}
We now introduce the pseudospectrum \cite{TrefethenEmbree2005} for a truncated Hamiltonian $\mathcal{H}_W$, $\sigma_{\epsilon}(\mathcal{H}_W)$. This is defined as
\be
\sigma_{\epsilon}(\mathcal{H}_W) = \{\omega \in \mathbb{C} \;\big|\; \|R(\omega;\mathcal{H}_W)\|_{E_\psi} \geq \epsilon^{-1}\}\,, \label{eq:pseudo_W}
\ee
where $R(\omega;\mathcal{H}_W)\equiv (\omega-\mathcal{H}_W)^{-1}$ is the resolvent, $\mathcal{H}_W$ is Hamiltonian for the subspace $W$ and the size of the perturbation $\epsilon$ is measured with respect to the energy norm $\|\cdot\|_{E_\psi}$. As we show explicitly in appendix \ref{app:pseudoW}
\be
\|R(\omega;\mathcal{H}_W)\|_{E_\psi} = \|R(\omega;H_W)\|_2\,, \label{eq:pseudoW}
\ee
where $H_W$ is defined in section \ref{sec:optimal}. $H_W$ is a finite dimensional matrix and is computed by knowing only a subset of the spectrum of $\mathcal{H}$, and thus \eqref{eq:pseudo_W} has the advantage of being a well defined quantity. This is in contrast with $\sigma_{\epsilon}(\mathcal{H})$, which when computed through numerical approximations do not converge to a continuum value in some cases~\cite{Jaramillo:2020tuu}. 

In figure \ref{fig:pspecW}~(left) we plot the sets $\sigma_\epsilon(\mathcal{H}_W)$ for several $\epsilon$, with $W$ containing only the first $M=10$ QNMs. This exhibits the usual characteristics anticipated for spectrally unstable systems. In addition to this, for our purposes, the contours protrude significantly enough into the upper half plane to imply transient growth through the Kreiss Matrix Theorem. This theorem arises by writing the resolvent as a Laplace transform of the time evolution operator, $i(\omega-\mathcal{H})^{-1}=\int_0^\infty e^{i\omega\tau}e^{-i\mathcal{H}\tau}d\tau$,\footnote{Convergent if $\im\,\omega>\im\,\lambda$ for any eigenvalue $\lambda$ of $\mathcal{H}$.} whose norm can then be straightforwardly bounded,\footnote{If $\im\,\omega>0\geq \im\,\lambda$ for any eigenvalue $\lambda$ of $\mathcal{H}$.}
\bea
\bigg\|\int_0^\infty e^{i\omega\tau}e^{-i\mathcal{H}\tau}d\tau\bigg\|&\leq& \int_0^\infty |e^{i\omega\tau}|\|e^{-i\mathcal{H}\tau}\|d\tau \nonumber\\
&\leq& \sup_{\tau\geq0}\|e^{-i\mathcal{H}\tau}\|\int_0^\infty e^{-\im\,\omega\tau}\,d\tau \nonumber\\
&=& \frac{1}{\im\,\omega}\sup_{\tau\geq0}\|e^{-i\mathcal{H}\tau}\| \label{eq:KMTbound}
\eea
yielding a result that relates the pseudospectrum contours to growth in time,
\be
\mathcal{K}(\mathcal{H}) \equiv \sup_{\im \,\omega >0}(\im\,\omega) \|R(\omega;\mathcal{H})\|\leq\sup_{\tau\geq0}\|e^{-i\mathcal{H}\tau}\|, \label{KMTuntruncated}
\ee
where $\mathcal{K}(\mathcal{H})$ is the Kreiss constant. Specialised to our truncated system and the energy norm, \eqref{KMTuntruncated} implies
\be
\mathcal{K}^2(\mathcal{H}_W) \leq  \sup_{\tau \geq 0}G_W(\tau), \label{KMT}
\ee
where
\be
\mathcal{K}(\mathcal{H}_W) = \sup_{\im \,\omega >0}(\im\,\omega) \|R(\omega;\mathcal{H}_W)\|_{E_\psi}= \sup_{\im \,\omega >0}(\im\,\omega) \|R(\omega;H_W)\|_2. \label{Kconstval}
\ee

Thus if $\mathcal{K}(\mathcal{H}_W) > 1$ there exist perturbations that exhibit energy growth. Focusing on $\re\,\omega=0.03$ (approximately the same real part as the fundamental QNM frequency) and moving upwards into the upper half plane, we track the value of $\im\,\omega\|R(\omega;H_W)\|_2$ against $\im\,\omega$; this is shown in figure \ref{fig:pspecW}~(right). The maximum of this curve then gives a lower bound on $\mathcal{K}(\mathcal{H}_W)$ through \eqref{Kconstval}.\footnote{The Kreiss constant is computed at the points of maximal protrusion in the upper half plane. Here, we instead obtain a bound by focusing on a fixed value of $\re\,\omega$.} Notably $\mathcal{K}(\mathcal{H}_W) > 1$ indeed, in agreement with the observation of transient growth through the construction of optimal perturbations in section \ref{sec:optimal}. The shaded area on figure~\ref{fig:GrowthFactor} denotes the region excluded  by the Kreiss Matrix Theorem~\eqref{KMT} for various Hamiltonian truncations with $\dim{W} = M$.

\begin{figure}[h!]
\centering
\includegraphics[width=0.5\columnwidth]{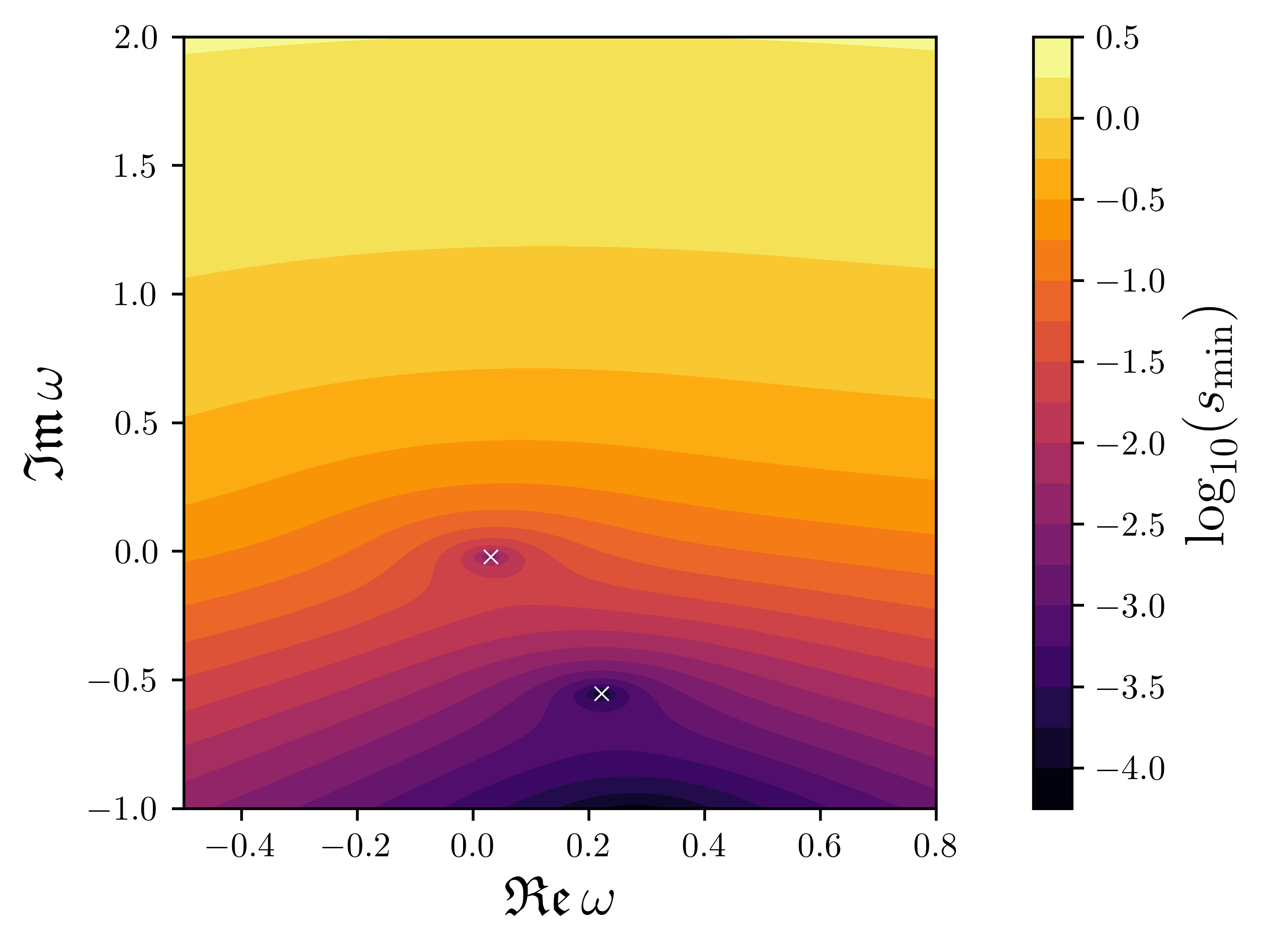}\qquad\includegraphics[width=0.43\columnwidth]{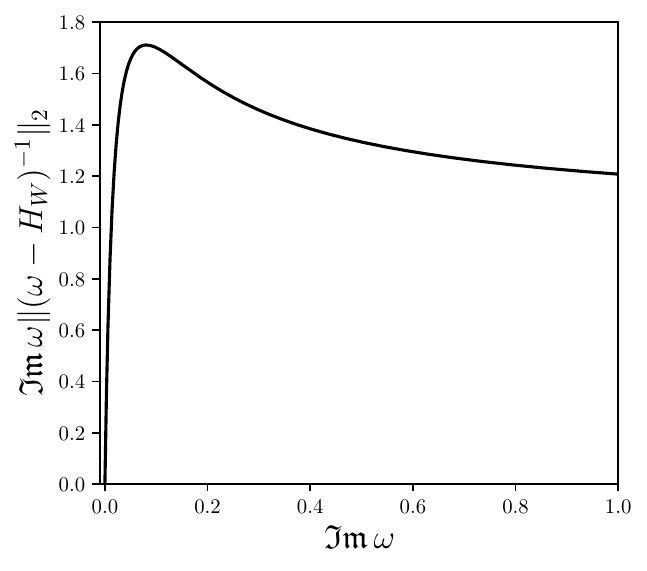}
\caption{\textbf{Left:} Pseudospectrum for the truncated Hamiltonian. Here, $s_{\min}(\omega-H_W)=\|(\omega-H_W)^{-1}\|^{-1}_{2}$ and we make use of the equivalent definition of the pseudospectrum $\sigma_{\epsilon}(\mathcal{H}_W) = \{\omega \in \mathbb{C} \;\big|\; s_{\text{min}}(\omega - H_W) \leq \epsilon \}$. Note the protrusion of pseudospectral contour lines in the upper half plane. \textbf{Right:} $\im \,\omega \|(\omega-H_W)^{-1}\|_2$ as a function of $\im\, \omega$ in the upper half plane for $\re\,\omega=0.03$ (approximately the real part of the frequency of the lowest lying QNM). This gives an indication of the depth of the protrusion of the pseudospectral contour line in the upper half plane, with the peak on this plot setting a lower bound for the Kreiss constant. The choice of parameters is $q=1$, $\mu=2.9$, $\vec{k}=0$ and $M=10$.}
\label{fig:pspecW}
\end{figure}

\section{A positivity theorem for QNM energies}\label{sec:positivity_thrm}
In the previous section, we have demonstrated that charged linear scalar perturbations on the RN-AdS$_4$ background can exhibit transient amplification in their energy via superradiance in the modally stable regime. This arises due to non-modal effects in the time evolution of sums of QNMs on a $\tau=0$ slice, ultimately coming from the non-normality of the system.\footnote{We have also checked that transient superradiance can still take place for generic choices of initial data. We discuss this matter further in section \ref{discussion}.} This is in contrast with the usual picture of superradiance in which individual scattered plane waves contain higher energy than the incident wave, either from infinity in asymptotically flat space \cite{PhysRevD.7.949} (see \cite{Brito:2015oca} for a review), or from sources at the conformal boundary of AdS \cite{Ishii:2022lwc}.

In this section we derive under which conditions QNMs exhibit (transient) superradiant amplification in our system, and, as a result, we arrive at a theorem on the positivity of QNM energies.

Consider the flux of the scalar field energy \eqref{Epsicons}. Working at linear level in $\psi$ translates into working at $\mathcal{O}(\epsilon)^2$ in \eqref{Epsicons} such that
\bea
\p_\tau E_\psi &=& -2 \int d^2\vec{x} |\p_\tau\psi|^2\bigg|_{z=1} + \mu \int  dz d^2\vec{x} \sqrt{-g} J^z \nonumber\\
&=& -2 \int d^2\vec{x} |\p_\tau\psi|^2\bigg|_{z=1} + \mu \int d^2\vec{x} \,\p_\tau Q_\psi, \label{eq:Epsi_flux}
\eea
where we have used \eqref{eq:StressTensor}, $F_{z\tau}=-\mu+\mathcal{O}(\epsilon)^2$ in the first line, and $\nabla_\mu J^\mu=0$ along with integration by parts in the second. Thus, upon plane wave decomposition \eqref{scalarplanewave}, considering a single $\vec{k}$-mode contribution (dropping the associated label) and for a single QNM with frequency $\omega$, $\chi(\tau,z)=e^{-i\omega\tau}Z(z)$, one has
\be
\label{eq:flux_1QNM}
\p_\tau E_{\psi}^{\text{\tiny QNM}}=-2e^{2 \im\,\omega\tau}\(|\omega|^2 + \mu q \re\,\omega\)|Z(1)|^2,
\ee
naturally giving rise to a condition that needs to be satisfied by $\omega$ for a single QNM to exhibit superradiance. Namely,
\be
\label{eq:superrCond}
\forall \tau \in \mathbb{R}:\p_\tau E_{\psi}^{\text{\tiny QNM}} >0 \,\Leftrightarrow \,\omega \in \left\{\tilde{\omega}\in\mathbb{C} \;\Big|\; \left(\re{\,\tilde{\omega}} + \frac{\mu q}{2}\right)^2 + (\im{\,\tilde{\omega}})^2 < \left(\frac{\mu q}{2}\right)^2\right\},
\ee
region in $\mathbb{C}$ corresponding to an open disk of radius $\mu q/2$ centred at $(-\mu q/2,0)$. This result leads to the following theorem

\begin{figure}[h!]
\centering
\includegraphics[width=0.5\columnwidth]{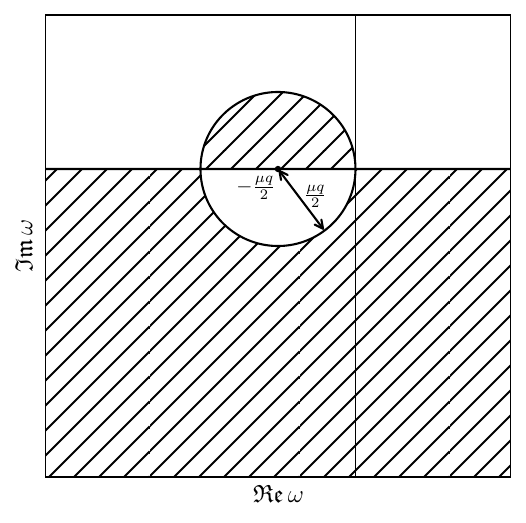}
\caption{The region $S$ in Theorem \ref{thm:energy}.}
\label{fig:EposRegion}
\end{figure}

\begin{theorem}\label{thm:energy}
Let $\xi(0,z)$ be a scalar QNM with charge $q$ on the RN-AdS$_4$ background \eqref{eq:RN} at chemical potential $\mu$. Let $\omega$ be its associated eigenfrequency, such that $\xi(\tau,z)=e^{-i\omega\tau}\xi(0,z)$. Then, for $\omega$ with $\im\,\omega\neq0$, its energy \eqref{eq:Epsi} $E_\psi[\xi(\tau,z)]>0 \;\forall \tau \in \mathbb{R}$ if and only if $\omega\in S = S_1\cup S_2$ (see figure \ref{fig:EposRegion}), where
\bea
S_1 &\equiv& \{\tilde{\omega}\in\mathbb{C} \;\big|\; \im\,\tilde{\omega}>0,\; \left(\re\,\tilde{\omega} + \frac{\mu q}{2}\right)^2+(\im\,\tilde{\omega})^2<\left(\frac{\mu q}{2}\right)^2\}, \\
S_2 &\equiv& \{\tilde{\omega}\in\mathbb{C} \;\big|\; \im\,\tilde{\omega}<0,\; \left(\re\,\tilde{\omega} + \frac{\mu q}{2}\right)^2+(\im\,\tilde{\omega})^2>\left(\frac{\mu q}{2}\right)^2\}.
\eea
\end{theorem}
    
\begin{proof}
From \eqref{eq:Epsi}, the time dependence of the energy of a single QNM $\xi(\tau,z)=e^{-i\omega\tau}\xi(0,z)$ is given by
\be
\label{eq:Eproof}
E_\psi[\xi(\tau,z)] = C e^{2\im\,\omega\, \tau},
\ee
with constant $C$.

\textit{$(\Rightarrow)\,$by contradiction}:
\begin{itemize}
    \item $\im\, \omega > 0$: Assume $\omega \in S^c$. Then \eqref{eq:flux_1QNM} implies $\p_\tau E_\psi[\xi(\tau,z)] \leq 0$ which by \eqref{eq:Eproof} implies $C \leq 0$ and thus $E_\psi[\xi(\tau,z)] \leq 0$ yielding a contradiction.
    \item $\im\, \omega < 0$: Assume $\omega \in S^c$. Then \eqref{eq:flux_1QNM} implies $\p_\tau E_\psi[\xi(\tau,z)] \geq 0$ which by \eqref{eq:Eproof} implies $C \leq 0$ and thus $E_\psi[\xi(\tau,z)] \leq 0$ yielding a contradiction.
\end{itemize}

\textit{$(\Leftarrow)$}:
\begin{itemize}
    \item $\im\, \omega > 0$: When $\omega \in S$ \eqref{eq:flux_1QNM} implies $\p_\tau E_\psi[\xi(\tau,z)] > 0$, which by \eqref{eq:Eproof} implies $C > 0$ and thus $E_\psi[\xi(\tau,z)] > 0$.
    \item $\im\, \omega < 0$: When $\omega \in S$ \eqref{eq:flux_1QNM} implies $\p_\tau E_\psi[\xi(\tau,z)] < 0$, which by \eqref{eq:Eproof} implies $C > 0$ and thus $E_\psi[\xi(\tau,z)] > 0$.
\end{itemize}
\end{proof}

Thus far, the analysis has been concerned with linear perturbations formed of a single QNM. We have found that, provided that its frequency $\omega$ satisfies the superradiant condition \eqref{eq:superrCond}, the associated energy is monotonically increasing, i.e. it is not a transient effect.

Let us now consider a linear perturbation constructed from a sum of 2 QNMs $\chi(\tau,z)=c_1 e^{-i\omega_1\tau}Z_1(z)+c_2e^{-i\omega_2\tau}Z_2(z)$, with flux \eqref{eq:Epsi_flux}
\bea
\p_\tau E_\psi=&-&|c_1|^22e^{2 \im\,\omega_1\tau}\(|\omega_1|^2 + \mu q \re\,\omega_1\)|Z_1(1)|^2 -|c_2|^22e^{2 \im\,\omega_2\tau}\(|\omega_2|^2 + \mu q \re\,\omega_2\)|Z_2(1)|^2 \nonumber \\
&-& 2 \re \left[c_1 \overline{c}_2 e^{-i(\omega_1 - \overline{\omega}_2)\tau}(2\omega_1\overline{\omega}_2+\mu q (\omega_1+\overline{\omega}_2))Z_1(1)\overline{Z}_2(1)\right].
\eea
Note that the first two terms above correspond exactly to the flux of each QNM alone, as seen in \eqref{eq:flux_1QNM}. However, also note the appearance of the third term mixing QNM$_1$ and QNM$_2$ -- a non-orthogonal term that opens the door for a transient period where $\p_\tau E_\psi>0$ even when neither $\omega_1$ or $\omega_2$ satisfy the superradiant condition \eqref{eq:superrCond}. This highlights how the phenomenon of transient superradiance studied in this paper is rooted in non-normal physics.

\section{Discussion}\label{discussion}

In this work we have shown that transient energy growth can occur for linear perturbations of black holes, even in situations where all QNMs are decaying.\footnote{Here, superradiance is absent only at the classical level, however note it was recently pointed out that quantum effects can also induce it \cite{Cartwright:2025fay}.} We focussed on the prototypical example of the holographic superconductor. This system shares many common features with the plane-Poiseulle flow in fluid dynamics, as listed in table \ref{tab:analogy}. In both systems, transient energy growth occurs by borrowing energy from a bath -- the charged black hole and the mean fluid flow, respectively. In the former it corresponds to a transient breaking of $U(1)$ symmetry, while in the latter, to the breaking of translational invariance. 

The superconducting instability can be seen as the classical field analogue of Schwinger pair production in AdS$_2$~\cite{Pioline:2005pf, Hartnoll:2011fn, Anninos:2019oka}. A natural question is therefore whether the transients that we have seen here are also rooted in AdS$_2$ physics. Indeed, in the case where the critical temperature of the holographic superconductor is very low and the QNM spectra appear similar to that of AdS$_2$, the transient growth is more pronounced. It would also be interesting to investigate the connection to the Aretakis instability \cite{Aretakis:2012ei} for this reason (see also \cite{Murata:2013daa, Gralla:2016sxp}).

We analysed the maximum energy growth through the construction of optimal perturbations. However, it does not appear to be the case that such perturbations are fine tuned. Preliminary results indicate that simply adding Gaussian initial data peaked near the horizon has similar energy dynamics. Thus, the large coefficient functions in a QNM expansion of optimal perturbations are likely reflective of the poor level of orthogonality displayed by these modes, rather than being indicative of fine tuning. Indeed, they are of the same order of magnitude as the excitation coefficients defined using the orthogonality relation of \cite{Green:2022htq}.

When is the transition between transient and modal behaviour? This question is particularly relevant when trying to extract QNMs from time-domain signals. Without prior knowledge of the transition point, this process could result in an attempt to measure a QNM amplitude from a transient signal. Such considerations are relevant in black hole ringdown, and specifically within the spectroscopy programme~\cite{Zhu:2023mzv}, as well as in the context of thermalisation of strongly coupled systems within AdS/CFT. 

The effects we discuss arise from dissipation via energy loss through the black hole horizon. One may wonder how this is reflected in the dynamics of one point functions in a dual theory. The dissipative nature of the horizon gives rise to damped QNMs and in linear response, CFT one point functions are given by a sum of these modes. In the bulk, this dissipation enabled us to construct states which take an arbitrarily long time to thermalise. In the boundary it is easy to see that a transiently growing one point function can also be arranged by carefully choosing the coefficients in a sum of decaying exponentials. These two effects are related through reading off QNM coefficients from bulk initial data.

In the astrophysical context, it is by now well established that superradiance can facilitate the detection of light bosonic degrees of freedom, relevant in searches for dark matter and physics beyond the Standard Model \cite{Arvanitaki:2009fg, Arvanitaki:2010sy, Pani:2012vp, Brito:2013wya, Herdeiro:2014goa}. Specifically, for spinning massive black holes, the spin-down rate can put constraints on the mass of the bosonic fields triggering superradiance. Traditional analysis, usually considered in the form of scattering monochromatic waves off the black hole, shows that only bosons with mass below a certain threshold can lead to superradiance. In this work we have shown that superradiance can also be seen transiently, meaning that one might see a spin-down of the black hole followed by a spin-up, in situations where the mass of the boson is not in the superradiant range.

A by-product of this work was the introduction of the truncated-Hamiltonian pseudospectrum. This quantity is numerically convergent. This is an advantage compared to the full-Hamiltonian pseudospectrum which is not, at least in hyperboloidal foliations. In null coordinates, the full pseudospectrum enjoys improved convergence properties~\cite{Boyanov:2023qqf}, and also when the energy norm is appropriately modified to include higher derivative terms that enforce a functional space of higher regularity~\cite{Warnick:2013hba}. The truncated-Hamiltonian pseudospectrum, also imposes restrictions on the functional space, since it only takes into account a subset of the spectrum. It is thus an alternative to existing approaches in the literature to UV regulate the pseudospectrum.

We also introduced a theorem for QNM energies, giving a rigorously determined region in the complex $\omega$ plane where QNMs with positive energies lie. It is possible that knowledge of this region can be used as an additional ingredient in analytic bootstrap approaches to QNM properties, for example by combining it with recently developed causality bounds \cite{Heller:2022ejw, Heller:2023jtd}, or knowledge of pole-skipping point locations \cite{Grozdanov:2017ajz, Blake:2017ris, Blake:2018leo}.

Finally, it is worth reiterating that the mechanism behind transient growth is linear. What remains to be seen is whether such growth can source and sustain interesting nonlinear effects.\footnote{See \cite{Bosch:2019anc} for an existing exploration of nonlinear superradiant dynamics in this model.} Thus, an important open question is whether there is a sustained period of out-of-equilibrium dynamics, akin to the transition to turbulence in shear flows. From the holographic perspective this would represent novel states of strongly interacting QFTs. We hope to report on this in future work.

\section*{Acknowledgements}
It is a pleasure to thank Mark Mezei, Alex Ratcliffe and Laura Sberna for discussions, and Roberto Emparan for comments on the draft.
B.W. would like to thank Ghent University and the organisers of `Quantum Dynamics at Ghent 2025' for their hospitality, where this work was also presented.
J.C. is supported by the Royal Society Research Grant RF/ERE/210267. 
C.P. is supported by a Royal Society -- Research Ireland University Research Fellowship via grant URF/R1/211027. 
B.W. is supported by a Royal Society University Research Fellowship and in part by the Science and Technology Facilities Council (Consolidated Grant ``New Frontiers In Particle Physics, Cosmology And Gravity'').

\appendix
\section{Calculations for the perturbed plane-Poiseuille flow} \label{app:PPF}
The incompressible Navier-Stokes equation is given by,
\be \label{NS}
(\partial_t + \vec{u}\cdot\vec{\nabla})\vec{u} = -\rho^{-1}\vec{\nabla} P + \nu\nabla^2\vec{u},
\ee
where $\vec{u}$ is the fluid 3-velocity, $P$ the pressure, $\rho$ a constant density and $\nu$ the constant kinematic viscosity. We consider plane-Poiseuille flow plus a stream function perturbation, $\Phi(t,x,y)$, and pressure perturbation $\delta P(t,x,y)$, viz.
\bea
u^x &=& U(y) +  \partial_y\Phi(t,x,y),\\
u^y &=& 0 - \partial_x \Phi(t,x,y),\\
u^z &=& 0,\\
P &=& -2\nu\rho x + \delta P(t,x,y),
\eea
where $U(y) = 1-y^2$. The perturbed flow is incompressible, $\vec{\nabla}\cdot \vec{u} = 0$, and we impose no-slip boundary conditions $\Phi(t, \pm 1) = \partial_y \Phi(t,\pm 1) = 0$. In order to correctly account for energy in the system we work to second order in perturbation theory,
\bea
\delta P &=& \epsilon P^{(1)}(t,x,y) + \epsilon^2 P^{(2)}(t,x,y) + O(\epsilon)^3\\
\Phi &=& \epsilon \Phi^{(1)}(t,x,y) + \epsilon^2 \Phi^{(2)}(t,x,y) + O(\epsilon)^3,
\eea
where $\epsilon$ is a formal parameter counting orders in the expansion. Let us consider each order in turn.

At $O(\epsilon)$ we take a real perturbation formed as follows:
\be
\Phi^{(1)}(t,x,y) = \phi(t,y) e^{i \alpha x} + \overline{\phi}(t,y) e^{-i \alpha x} \label{Phiansatz}
\ee
with wavenumber $\alpha \neq 0$. $P^{(1)}$ is determined algebraically by the $x$-component of \eqref{NS} at order $\epsilon$. 
The $y$ component of \eqref{NS}, after inserting $P^{(1)}$, gives the Orr-Sommerfeld equation \eqref{OSeq}, where the Orr-Sommerfeld operator is
\be \label{OSoperator}
O_\text{OS} = -\Delta_2^{-1}
\left[(i\alpha\, Re)^{-1} \Delta_2^2 - U(y)\Delta_2 +U''(y)\right],
\ee
with Reynolds number $Re = \nu^{-1}$, and where $\Delta_2 = \partial_y^2-\alpha^2$ is the spatial Laplacian for the $x,y$-plane.

At $O(\epsilon)^2$ we have perturbations sourced by $O(\epsilon)$ terms. These take the form,
\be
\Phi^{(2)}(t,x,y) = \delta\phi_0(t,y) + \sum_\pm \delta\phi_{\pm\alpha}(t,y) e^{\pm i 2\alpha x}.
\ee
The zero-momentum piece $\partial_y\delta\phi_0$ obeys the following diffusion equation,
\be
\left(\partial_t - \nu \partial_y^2\right)\partial_y \delta \phi_0 = i \alpha \partial_y\left(\phi \partial_y \overline{\phi} - \overline{\phi} \partial_y \phi\right), \label{deltaphieq}
\ee
with viscosity $\nu$ serving as the diffusivity, and a current source term coming from the $O(\epsilon)$ perturbations. $P^{(2)}$, $\delta\phi_{\pm\alpha}(t,y)$ are determined by other differential equations but we will not need them here.

To $O(\epsilon)^2$ the energy of the perturbed flow evaluates to
\bea
E = \int_{-1}^1 dy \int dx \, \overline{u}\cdot u = \text{vol}_x \bigg(\frac{16}{15} &+& 2\epsilon^2 \int_{-1}^1 \left(|\partial_y \phi|^2 + \alpha^2 |\phi|^2 \right)dy \\
&+& 2\epsilon^2 \int_{-1}^1 (1-y^2)\,\partial_y\delta\phi_0 \,dy + O(\epsilon)^3\bigg).
\eea
We isolate the two contributions at order $\epsilon^2$, which we write as
\bea
E_\phi &=& \int_{-1}^1 \left(|\partial_y \phi|^2 + \alpha^2 |\phi|^2 \right)dy,\\
\delta E_U &=& \int_{-1}^1 (1-y^2)\, \partial_y\delta\phi_0\, dy.
\eea
$E_{\phi}$ matches the energy norm $||\phi||^2_H$ in \cite{Reddy93} equation (A.6).\footnote{The second term in \cite{Reddy93} (A.6) appears to have a typo, reading $|\phi|$ instead of $|\phi|^2$. The factors of $2$ in our energies come from a different perturbation strength in \eqref{Phiansatz}.}  Note that $E \geq 0$ and $E_\phi \geq 0$ but no such restriction exists for $\delta E_U$. One can show that
\be
\partial_t\left(E_\phi + \delta E_U\right) = -\frac{2}{Re}\int_{-1}^1\left(|\partial_y^2 \phi|^2 + 2 \alpha^2 |\partial_y \phi|^2 + \alpha^4 |\phi|^2\right) dy \leq 0. \label{dEviscous}
\ee
Thus, the total energy at order $\epsilon^2$, $E_\phi + \delta E_U$, is strictly non-increasing, and can decrease due to viscous dissipative effects. However, $\partial_t E_\phi$ alone contains a non-viscous bulk term which is not sign-definite and this is where transient growth can appear.

To see the behaviour of $E_\phi$ in practise, we consider the growth curve example given in \cite{Reddy93} at parameter choices $\alpha = 1$, $\nu^{-1} = Re = 5000$. We use $200$ Chebyshev grid points and the subspace $W$ is formed of $30$ eigenfunctions, whose eigenvalues are plotted in figure \ref{fig:OSspec}. We first evaluate $E_\phi$ for an optimal perturbation corresponding to $E_\phi(t=0) = 1$ and $E_\phi$ which achieves the peak of the growth curve, using initial data corresponding to the right principal singular eigenvector of the truncated time evolution operator. From figure \ref{fig:OSEEE} we see that $E_\phi$ displays transient growth, and then modally returns to zero over a time set by the longest lived QNM, which in this case is quite long: $(-\im{\,\omega})^{-1} \simeq 571$. This is the result established in \cite{Reddy93}. Unlike the growth curve for transient superradiance (see the discussion surrounding figure \ref{fig:GrowthFactor}), here increasing the dimension of $W$ does not lead to an increased peak of the curve.

To evaluate $\delta E_U$ we solve for $\partial_y\delta\phi_0$ in \eqref{deltaphieq} using the Crank-Nicolson method, starting from $\delta E_U = \delta\phi_0(t=0) = 0$. The result is also shown in figure \ref{fig:OSEEE}, showing a decrease $\delta E_U$. This is the counterpart to the result of \cite{Reddy93}, i.e. the missing piece that ensures that the total $E_\phi + \delta E_U$ cannot increase (it just decreases due to viscous dissipation). Thus one may interpret this transient process as an interaction between the quadratic zero-momentum mode $\delta\phi_0$ and the two linear momentum modes $\phi$, $\overline{\phi}$ -- a coupling which exists only due to the background flow $U(y)$.

\section{Coefficients appearing in energy and charge integrals}\label{app:wpq}

The coefficient functions appearing in \eqref{energyintegral} are given by,
\bea
w_E(z)&=&z^4\tilde{w}_E(z),\nonumber\\
p_E(z)&=&z^4\tilde{p}_E(z),\nonumber\\
q_E(z)&=&z^4\tilde{q}_E(z)-2\p_z\left(\tilde{p}_E(z)z\right)z^2+z^2 \vec{k}^2, \nonumber\\
\alpha_1(z)&=&z^4 \tilde{\alpha}_1(z), \nonumber\\
\alpha_2(z)&=&z^4\tilde{\alpha}_2(z),
\eea
where
\bea
\tilde{w}_E(z) &=& -\frac{\left((4-4z^3+z^4\,\mu^2-z^3\mu^2)^2 h'^2-16\right)}{4 z^2 \left(\mu ^2 z^4-\left(\mu ^2+4\right) z^3+4\right)},\nonumber \\
\tilde{p}_E(z) &=& \frac{\mu ^2 z^2}{4}+\frac{1}{z^2}-\frac{1}{4} \left(\mu ^2+4\right) z,\nonumber \\
\tilde{q}_E(z) &=& \frac{2 \left(\mu ^2 \left(2 q^2-1\right) z^3+z^2 \left(4-2 \mu ^2 q^2\right)+4 z+4\right)}{z^4 \left(\mu ^2 z^3-4 z^2-4 z-4\right)}, \nonumber\\
\tilde{\alpha}_1(z) &=& -i\frac{\mu  q (1-z) \left((4-4z^3+z^4\,\mu^2-z^3\mu^2)^2 h'^2-16\right)}{4 z^2 \left(\mu ^2 z^4-\left(\mu ^2+4\right) z^3+4\right)}, \nonumber\\
\tilde{\alpha}_2(z) &=& -i\frac{\mu  q (z-1)^2\left(4(1+z+z^2)-z^3\mu^2\right) h'}{4z^2}\,.
\eea
Note that to arrive at \eqref{energyintegral} we have removed the boundary term $\int_0^1 dz \,\p_z(\tilde{p}_E(z)2z^3\overline{\chi}\chi)=0$.

The coefficient functions appearing in \eqref{chargeintegral} are given by,
\bea
w_Q(z) &=& z^4 \tilde{w}_Q(z),\nonumber \\
p_Q(z) &=& z^4 \tilde{p}_Q(z), \nonumber\\
q_Q(z) &=& z^4 \tilde{q}_Q(z) + 2z^3 \tilde{p}_Q(z).
\eea
where 
\bea 
\tilde{w}_Q(z) &=&-i\,q\frac{\left(-16 + \left(4 + z^4 \mu^2 - z^3 (4 + \mu^2)\right)^2 h'^2\right)}{4 z^2 \left(4 + z^4 \mu^2 - z^3 (4 + \mu^2)\right)},\nonumber\\
\tilde{p}_Q(z) &=&-i\,q \frac{\left(4 + z^3 \left(-4 + (-1 + z) \mu^2\right)\right) h'}{4 z^2}, \nonumber\\
\tilde{q}_Q(z) &=& \frac{4 \mu  q^2}{-\mu ^2 z^5+4 z^4+4 z^3+4 z^2}.
\eea
Note that these functions are written in terms of $h'(z)$, where $h(z)$ is given in \eqref{hchoice} for our specific case.

\section{Proof of \eqref{eq:pseudoW}}\label{app:pseudoW}
Given the formal expansion,
\be
(\omega-\mathcal{H}_W)^{-1}=\frac{\mathbb{I}}{\omega}+\frac{\mathcal{H}_W}{\omega ^2}+\frac{\mathcal{H}_{W}^2}{\omega ^3}+\ldots\,,
\ee
it can be shown that
\bea
\|(\omega -\mathcal{H}_W)^{-1} \xi(0,z)\|_{E_\psi}^{2} &=& \langle (\omega -\mathcal{H}_W)^{-1} \xi(0,z) , (\omega -\mathcal{H}_W)^{-1} \xi(0,z) \rangle \nonumber\\
&=& \sum_{n,m=1}^{M}c_{n}^*c_m \langle (\omega -\mathcal{H}_W)^{-1} \tilde{\xi}_n , (\omega -\mathcal{H}_W)^{-1} \tilde{\xi}_m \rangle \nonumber\\
&=& \sum_{n,m=1}^{M}c_{n}^*c_m \[(\omega-\omega_{n})^{-1}\]^*(\omega-\omega_{m})^{-1} \langle \tilde{\xi}_n , \tilde{\xi}_m \rangle\nonumber\\
&=& \sum_{n,m=1}^{M}\sum_{j,k=1}^{M}\[(\omega-\omega_{n})^{-1} c_n\]^* \[(\omega-\omega_{m})^{-1} c_m\] ((U_W)_{jn})^* (U_W)_{km} \langle \zeta_j, \zeta_k \rangle \nonumber\\
&=& \sum_{n,m=1}^{M}\sum_{j=1}^{M} \[(U_W)_{jn} (\omega-\omega_{n})^{-1} c_n\]^* \[(U_W)_{jm} (\omega-\omega_{m})^{-1} c_m \] \nonumber\\
&=& \[U_W (\omega-D_W)^{-1} \vec{c}\]^*\[U_W (\omega-D_W)^{-1} \vec{c}\] \nonumber\\
&=& \|U_W (\omega-D_W)^{-1} \Vec{c}\,\|_{2}^{2} \nonumber\\
&=& \|U_W (\omega-D_W)^{-1} U_{W}^{-1} \,\Vec{d}\,\|_{2}^{2} \nonumber\\
&=& \|(\omega-H_W)^{-1}\,\Vec{d}\,\|_{2}^{2}.
\eea
Similarly,
\bea
\|\xi(0,z)\|_{E_\psi}^2 &=& \langle \xi(0,z) , \xi(0,z) \rangle \nonumber\\
&=& \sum_{n,m=1}^{M}c_{n}^*c_m \langle \tilde{\xi}_n , \tilde{\xi}_m \rangle\nonumber\\ 
&=& \sum_{n,m=1}^{M}\sum_{j,k=1}^{M}c_n^* c_m ((U_W)_{jn})^* (U_W)_{km} \langle \zeta_j, \zeta_k \rangle \nonumber\\
&=& \(U_W \vec{c}\)^*\(U_W \vec{c}\)\nonumber\\ 
&=& \|U_W \Vec{c}\,\|_{2}^{2}\nonumber\\ 
&=& \|\Vec{d}\,\|_{2}^{2},
\eea
such that
\bea
\|(\omega-\mathcal{H}_W)^{-1}\|_{E_\psi}^{2}&=&\sup_{\xi(0,z) \in W} \frac{\|(\omega-\mathcal{H}_W)^{-1} \xi(0,z)\|_{E_\psi}^{2}}{\|\xi(0,z)\|_{E_\psi}^{2}}\nonumber\\& =& \max_{\vec{d}\in \mathbb{C}^M} \frac{\|(\omega-H_W)^{-1}\,\Vec{d}\,\|_{2}^{2}}{\|\Vec{d}\,\|_{2}^{2}} \nonumber\\
&=& \|(\omega-H_W)^{-1}\|_{2}^{2}.
\eea

In the above, we have used the expansion~\eqref{eq:expansion}, the change of basis $\tilde{\xi}_n = \sum_{m=1}^{M} (U_W)_{mn} \zeta_m$,  and the orthonormality relation $\langle \zeta_j , \zeta_k \rangle = \delta_{jk}$.

\bibliographystyle{ytphys}
\bibliography{refs}

\end{document}